\pdfoutput=1
\documentclass[10pt,pra,aps,superscriptaddress,nofootinbib,twocolumn,longbibliography]{revtex4-1}
\usepackage{amsmath,bbm}
\usepackage{amsthm}
\usepackage{amsmath}
\usepackage{latexsym}
\usepackage{amssymb}
\usepackage{float}
\usepackage{graphicx}           
\usepackage{color}
\usepackage{xcolor}
\usepackage{mathpazo}
\usepackage{comment}
\usepackage{enumitem}
\usepackage{multirow}
\usepackage{setspace}
\usepackage{subcaption}
\usepackage[colorlinks=true,linkcolor=blue,citecolor=magenta,urlcolor=blue]{hyperref}
\usepackage{svg}

\newcommand{\be}{\begin{equation}}
\newcommand{\ee}{\end{equation}}
\newcommand{\bea}{\begin{eqnarray}}
\newcommand{\eea}{\end{eqnarray}}
\def\squareforqed{\hbox{\rlap{$\sqcap$}$\sqcup$}}
\def\qed{\ifmmode\squareforqed\else{\unskip\nobreak\hfil
\penalty50\hskip1em\null\nobreak\hfil\squareforqed
\parfillskip=0pt\finalhyphendemerits=0\endgraf}\fi}
\def\endenv{\ifmmode\;\else{\unskip\nobreak\hfil
\penalty50\hskip1em\null\nobreak\hfil\;
\parfillskip=0pt\finalhyphendemerits=0\endgraf}\fi}

\newcommand{\tr}{\text{Tr}}
\newcommand{\I}{\mathbbm{I}}

\newcommand{\ket}[1]{|#1\rangle}
\newcommand{\bra}[1]{\langle#1|}

\makeatletter
\newtheorem*{rep@theorem}{\rep@title}
\newcommand{\newreptheorem}[2]{%
\newenvironment{rep#1}[1]{%
 \def\rep@title{#2 \ref{##1}}%
 \begin{rep@theorem}}%
 {\end{rep@theorem}}}
\makeatother

\newtheorem{thm}{Theorem}
\newreptheorem{thm}{Theorem}
\newtheorem{lemma}{Lemma}

\begin{document}

\title{Single shot distinguishability of noisy quantum channels}

\author{Satyaki Manna}
\email{mannasatyaki@gmail.com}
\affiliation{Department of Physics, School of Basic Sciences, Indian Institute of Technology Bhubaneswar, Odisha 752050, India}

\begin{abstract}
    Among the intriguing features of quantum theory, the problem of distinguishing quantum channels is of fundamental interest. In this paper, we focus on the single-shot discrimination of two noisy quantum channels using two distinct classes of probes: single-system (product) probes and entangled probes. Our aim is to identify optimal probing state for specific discrimination tasks and to analyze the necessity and role of entanglement in enhancing channel distinguishability.
 We show that maximally entangled probes are optimal for discriminating two qubit depolarizing channels, with any nonzero entanglement providing an advantage over single-system probes. In contrast, for dephasing channels in arbitrary dimensions, we prove that single-system probe can be optimal and that entanglement offers no improvement, even when the dephasing unitary is generalized.
For qubit amplitude-damping channels, we identify distinct noise-dependent regimes in which either single-system probe outperforms maximally entangled probes and vice-versa. Moreover, we demonstrate that non-maximally entangled probes can act as the optimum probe if the noise parameters restricted to certain values in this task. We also present examples of noisy unitary channels for which discrimination is possible using non-maximally entangled probe, while both single-systems and maximally entangled probes fail. We introduce another class of noisy unitary channels for which perfect discrimination is achievable with a single system, while maximally entangled probes are insufficient. Finally, we show that two erasure channels can be optimally discriminated using any pure single-system probe, with no advantage gained from entanglement.
\end{abstract}

\maketitle


\section{Introduction}
Distinguishability of different physical processes has long been an important line of investigation for understanding the subtlety of quantum theory and, more broadly, the physical world. In classical theory, the notion of distinguishability is straightforward. In contrast, quantum theory—characterized by several counterintuitive and nonclassical features—renders this concept significantly more peculiar. In general, distinguishability refers to the task of identifying which specific process has occurred from a known set of possible processes. From a foundational perspective, this notion is particularly useful for commenting on the ontological status and reality of quantum states\cite{leifer2014,barrett14,chaturvedi2021,Chaturvedi2020,Pusey_2012,ray2025,ray2024}. Beyond its foundational implications, distinguishability has numerous applications in quantum information processing, quantum cryptography\cite{Bae_2015,JánosABergou_2007}, and quantum advantageous privacy-preserving communication tasks\cite{Manna,pandit}.

Since the inception of this kind of research, the distinguishability of quantum states has been the most extensively studied topic\cite{Hellstrom, benett, vedral, Halder,JánosABergou_2007,Bae_2015}. More recently, the distinguishability of quantum channels has emerged as a rich and comprehensive field of study \cite{watrous05,npj,njp2021,acin2001,duan2009,harrow,manna3,manna4,manna2025,feng,Manna_2}, albeit with a considerably more complex structure compared to state discrimination. Nevertheless, most existing studies refrain from commenting on the optimal choice of probes for a given set of channels, except in the case of simple channels\cite{GMauroD’Ariano_2002,bsxv-q9x7,feng}. Another important research direction began with the seminal work of Kitaev \cite{AYuKitaev_1997}, who demonstrated that an entanglement-assisted probe–ancilla strategy can sometimes provide an advantage in discriminating quantum channels. Over the years, the advantage of entangled probes over product probes in distinguishing noiseless quantum channels has been extensively explored \cite{bsxv-q9x7,sacchi,sacchi2}. However, analogous investigations for noisy quantum channels remain largely unexplored. Another relevant hierarchy of probes is given by maximally and non-maximally entangled states; yet this distinction of amount of entanglement in discrimination of channels has not been investigated in sufficient depth.

Noisy quantum channels play a central role in quantum information theory, as they provide realistic models of physical processes in which interactions with the environment and decoherence are unavoidable. In quantum communication, such channels are employed to evaluate channel capacities \cite{schumacher,lloyd,schumacher2} and transmission fidelities\cite{oh}, while in quantum cryptography they are crucial for analyzing eavesdropping strategies and establishing security bounds under realistic conditions\cite{gisin,Bilal,wehner}. Noisy quantum channels have also found extensive applications in quantum metrology, particularly in the study of noise-assisted quantum advantages\cite{Gorecki,falaye2017,Rossi}. In addition, they play a vital role in quantum process tomography\cite{Kunjummen,Bouchard2019}, randomized benchmarking\cite{helsen}, and gate-set tomography\cite{vinas}, where they are used to characterize and validate the performance of quantum devices.

In this work, our primary objective is to identify the optimal probe for single-shot discrimination of noisy quantum channels by examining noisy channel distinguishability under restricted probe resources. In this context, we also comment on the necessity of an auxiliary system for this task. Along this, the hierarchy of maximally and non-maximally entangled probe is also scrutinized. To commence, we briefly review the definition of quantum channels and the general protocol for discriminating a known set of channels. We then describe two discrimination scenarios: one employing a single-system probe and the other using an entangled probe. At this stage, it is worth noting that the use of a product-system probe is equivalent to deploying a single-system probe. Subsequently, we proceed to our results.
Our analysis focuses on the distinguishability of five types of noisy channels, namely the depolarizing channel, dephasing channel, amplitude damping channel, noisy unitary channel, and erasure channel. We find that for the discrimination of two qubit depolarizing channels, maximally entangled probes are optimal, with all maximally entangled states performing equivalently. In addition to that, we show that even an infinitesimal amount of entanglement assistance provides an advantage over single-system probes.
In contrast, the case of dephasing channels yields a counterintuitive result. We prove that a single-system probe is sufficient for the discrimination of two $d$-dimensional dephasing channels, and any amount of entanglement does not enhance distinguishability. This conclusion remains unchanged when the dephasing unitary is replaced by an arbitrary unitary, indicating that the optimal probe structure is preserved under such generalization.
For two amplitude damping channels acting on $\mathbb{C}^2$, we identify two distinct regimes depending on the noise parameters of both the channels. We demonstrate that there exist channel pairs for which a single-system probe outperforms a maximally entangled probe, as well as pairs for which the opposite holds. Remarkably, we also observe that non-maximally entangled probes are the optimum probes in this task if  noise parameters of two channels satisfy a certain relation.
We then investigate noisy unitary channels and present two distinct sets of mixed unitary channels. For the first set, perfect discrimination is possible using a non-maximally entangled probe, while both single-system and maximally entangled probes fail. For the second set, discrimination is achievable using single system but not with a maximally entangled probe. Finally, we show that two erasure channels can be optimally discriminated using any single-system probe, and that entanglement does not offer any extra advantage. Remarkably, the distinguishability of two erasure channels is independent of the probing state, single or entangled.

The paper is structured as follows. The next section begins with a review of quantum channel discrimination, formulating two scenarios based on the choice of probes. In the subsequent section, we present our main results. Finally, we summarize our findings and outline possible directions for future research.
\section{Discrimination of quantum channels}
A quantum channel transforms an arbitrary density matrix to another density matrix. Mathematically, a quantum channel is described as a completely positive trace preserving (CPTP) maps $\mathcal{N}:\mathcal{L}(\mathcal{H}_A)\rightarrow\mathcal{L}(\mathcal{H}_B)$. Any map of this kind can be written as follows:
\be
\mathcal{N}(\rho)=\sum_{i=0}^{d-1} V_i\rho V_i^\dagger,
\ee
where $\rho\in\mathcal{L}(\mathcal{H}_A)$ and $V_i\in\mathcal{L}(\mathcal{H}_A,\mathcal{H}_B)$, $\forall i\in\{0,\cdots,d-1\}$. The operators $V_i$ are known as Krauss operators. The trace preserving property supplements the relation $\sum_i V_i^\dagger V_i=\I$. In this paper, we are going to deal with noisy channels. Noisy channels are non-unitary CPTP maps, that means they need more than one krauss operators to define them.

We consider a known set of $q$ channels $\{\mathcal{N}_{i|x}\}_x$ acting on a $d$-dimensional quantum state and $x\in\{1,\cdots,q\}$, where $\sum_i \mathcal{N}_{i|x}^\dagger \mathcal{N}_{i|x}=\I$, for each $x$. We assume that all the channels are sampled from a probability distribution $\{p_x\}_x$, i.e., $p_x>0$ and $\sum_x p_x=1$. To distinguish the channels, the unitary device is fed with a known quantum state, product or entangled and the device carries out one of $n$ transformations due to the corresponding quantum channels. After this process, the device gives a evolved state as the output. Therefore, one can perform any measurement on the evolved state and this measurement can be optimized such that the distinguishability of these evolved states will be maximum. Let us describe this optimal measurement by a set of POVM elements $\{M_{b}\}_b$, where $b\in\{1,\cdots,q\}$ and $\sum_b M_{b}=\I$. The protocol is successful in distinguishing the channels if $b$ is same as $x$. Any classical post processing of outcome $b$ can be included in the measurement $\{M_{b}\}_b$. Depending on the initial probing state, we can formulate two situations which are described in details in the next subsection.
\subsection{With Single System}
At first, the channel device is fed a single system $\rho$. After the device applying any of the channels from the set $\{\mathcal{N}_{i|x}\}_x$ , the output state will be $\mathcal{N}_{i|x}(\rho)=\sum_i \mathcal{N}_{i|x}\rho\mathcal{N}_{i|x}^\dagger$. After performing the measurement $\{M_b\}_b$, the distinguishability of the set of channels becomes the distinguishability of the set of evolved states. So distinguishability in this scenario, denoted by $\mathcal{D}_S$, is defined as,
\bea\label{D_S}
&&\mathcal{D}_S\left[\{\mathcal{N}_{i|x}\}_x,\{p_x\}_x\right]\nonumber\\
&=& \max_{\rho} \sum_x p_xp(b=x|x)\nonumber\\
&=& \max_{{\rho},\{N_b\}}\sum_x p_x\tr\left((\sum_i \mathcal{N}_{i|x}\rho\mathcal{N}_{i|x}^\dagger)M_{b=x}\right)\nonumber\\
&=& \max_{\rho} \mathcal{DS}\left[\left\{\sum_i \mathcal{N}_{i|x}\rho\mathcal{N}_{i|x}^\dagger\right\}_x,\{p_x\}_x\right].\nonumber
\eea
$\mathcal{DS}[\{\sigma_k\}_k,\{q_k\}_k]$ denotes the distinguishability of the set of states $\{\sigma_k\}_k$, sampled from the probability distribution $\{q_k\}_k$. This expression reduced as following by Hellstrom's formula \cite{Hellstrom} when we want to know the distinguishability of two channels $\mathcal{N}_{i|1}$ and $\mathcal{N}_{i|2}$:
\bea
&&\mathcal{D}_S\left[\{\mathcal{N}_{i|x}\}_{x=1}^2,\{p_x\}_{x=1}^2\right]\nonumber\\
&=& \max_{\rho} \mathcal{DS}\left[\left\{\sum_i \mathcal{N}_{i|x}\rho\mathcal{N}_{i|x}^\dagger\right\}_x,\{p_x\}_x\right].\nonumber\\
&=& \max_{\rho}\frac12\left(1+||p_1 \sum_i \mathcal{N}_{i|1}\rho\mathcal{N}_{i|1}^\dagger-p_2 \sum_i \mathcal{N}_{i|2}\rho\mathcal{N}_{i|2}^\dagger||\right).\nonumber\\
\eea
If the channels are sampled from equal probability distribution, i.e., $p_1=p_2=1/2$, the above expression reduces to, 
\bea\label{D_S1}
&&\mathcal{D}_S\left[\{\mathcal{N}_{i|x}\}_{x=1}^2,\{1/2,1/2\}\right]\nonumber\\
&=&\frac12+\frac14\max_{\rho}|| \sum_i \mathcal{N}_{i|1}\rho\mathcal{N}_{i|1}^\dagger- \sum_i \mathcal{N}_{i|2}\rho\mathcal{N}_{i|2}^\dagger||,
\eea
where $||.||$ (trace norm) denotes the sum of the absolute values of the eigenvalues. By convexity of trace norm, the maximum value of $|| \sum_i \mathcal{N}_{i|1}\rho\mathcal{N}_{i|1}^\dagger- \sum_i \mathcal{N}_{i|2}\rho\mathcal{N}_{i|2}^\dagger||$ is achieved when $\rho$ is the pure input state.
\subsection{With Entangled System}
Now, the unitary device is fed with a $d\otimes d'$ entangled state $\rho_{AB}$ and the device applies one of the channels from the set $\{\mathcal{N}_{i|x}\}_x$ on the part $A$ of the entangled state. The evolved state will be $\mathcal{N}_{i|x}(\rho_{AB})=\sum_i (\mathcal{N}_{i|x}\otimes\I)\rho_{AB}(\mathcal{N}_{i|x}^\dagger\otimes\I)$. Then one can perform the optimal measurement of dimension $dd'$. Similarly, the distinguishability of the channels in this scenario, denoted as $\mathcal{D}_E$, can be written as,
\bea\label{D_E}
&&\mathcal{D}_E\left[\{\mathcal{N}_{i|x}\}_x,\{p_x\}_x\right]\nonumber\\
&=& \max_{\rho_{AB}} \sum_x p_xp(b=x|x)\nonumber\\
&=& \max_{\rho_{AB},\{N_b\}} \sum_x p_x \tr\left[(\sum_i \mathcal{N}_{i|x}\otimes\I)\rho_{AB}(\sum_i \mathcal{N}_{i|x}^\dagger\otimes\I)M_{b=x}\right]\nonumber\\
&=& \max_{\rho_{AB}} \mathcal{DS}\left[\left\{(\sum_i \mathcal{N}_{i|x}\otimes\I)\rho_{AB}(\sum_i \mathcal{N}_{i|x}^\dagger\otimes\I)\right\}_x,\{p_x\}_x\right].\nonumber\\
\eea
From the consequence of convexity of trace norm and Schmidt decomposition, it can be proved that it is sufficient to take $d=d'$\cite{2016}.

Similarly, for two equiprobable channels, the distinguishability reads,
\bea
&&\mathcal{D}_E\left[\{\mathcal{N}_{i|x}\}_{x=1}^2,\{1/2,1/2\}\right]\nonumber\\
&=& \frac12+\frac14\max_{\rho_{AB}}|| \sum_i (\mathcal{N}_{i|1}\otimes\I)\rho_{AB}(\mathcal{N}_{i|1}^\dagger\otimes\I)-\nonumber\\
&&\sum_i (\mathcal{N}_{i|2}\otimes\I)\rho_{AB}(\mathcal{N}_{i|2}^\dagger\otimes\I)||.
\eea
Suppose the input state is a product state, that means $\rho_{AB}=\rho_A\otimes\rho_B$. Then the above equation can be reduced as,
\bea
&&\mathcal{D}_P\left[\{\mathcal{N}_{i|x}\}_{x=1}^2,\{1/2,1/2\}\right]\nonumber\\
&=& \frac12+\frac14\max_{\rho_{AB}}|| \Bigg(\sum_i \mathcal{N}_{i|1}\rho_{A}\mathcal{N}_{i|1}^\dagger
-\sum_i \mathcal{N}_{i|2}\rho_{A}\mathcal{N}_{i|2}^\dagger\Bigg)\otimes\rho_B||\nonumber\\
&=& \frac12+\frac14\max_{\rho_{AB}}|| \Bigg(\sum_i \mathcal{N}_{i|1}\rho_{A}\mathcal{N}_{i|1}^\dagger-\sum_i \mathcal{N}_{i|2}\rho_{A}\mathcal{N}_{i|2}^\dagger\Bigg)||.||\rho_B||\nonumber\\
&=& \frac12+\frac14\max_{\rho_{AB}}|| \Bigg(\sum_i \mathcal{N}_{i|1}\rho_{A}\mathcal{N}_{i|1}^\dagger-\sum_i \mathcal{N}_{i|2}\rho_{A}\mathcal{N}_{i|2}^\dagger\Bigg)||,
\eea
which is same as \eqref{D_S1}. That indicates product probes act like single system probe in this discrimination task. It is needless to say that this inference is true for any number of channels sampled from any probability distribution because subsystem B does not change for any operation. The optimum measurement can not generate any extra information if it measures the subsystem B. Only the subsystem A is responsible for the discrimination.
\section{Results}
In this section, we exhibit all our results in different subsections for each different channels. Throughout the paper,
we denote identity operator in dimension $d$ as $\I_d$.
\subsection{Depolarizing Channels}
The map of a depolarizing channel $(\Phi)$ can be written as,
\be
\Phi_q(\rho)=q\rho+(1-q)\frac{\I_d}{d};
\ee
where $\rho\in\mathbbm{C}^d$ and $q\in(0,1)$. One can interpret this channel such that it assumes we forget the input bit with some probability and it replaces it with maximally mixed state. In this subsection, we are interested in the distinguishability of two depolarizing channels in two scenarios, one with single system probe and the other with entangled system probe. At first, we analyze the distinguishability using single system. 
\begin{lemma}\label{th1}
    All pure single systems are equivalent probes for distinguishing two depolarizing channels.
\end{lemma}
\begin{proof}
    As \eqref{D_S1} suggests, we need to calculate the quantity
$||\Phi_{q_1}(\rho)-\Phi_{q_2}(\rho)||$, where the channels are given by
$\Phi_{q_1}(\rho)=q_1\rho+(1-q_1)\I_d/d$ and
$\Phi_{q_2}(\rho)=q_2\rho+(1-q_2)\I_d/d$.
Substituting the explicit forms of the channels into the trace norm, we obtain
\bea\label{phi}
||\Phi_{q_1}(\rho)-\Phi_{q_2}(\rho)|| &=& |q_1-q_2|.||\rho-\I_d/d||.
\eea

Since $\rho$ is a pure state, we may write $\rho=\ket{a}\bra{a}$.
Expanding $\ket{a}=\sum_i\alpha_i\ket{i}$, we have
\be
\rho-\I_d/d
=
\sum_{i,j}\left(\alpha_i\alpha_j^*\ket{i}\bra{j}
-\frac1d\ket{i}\bra{i}\right)
=\mathbf{A}.
\ee
Let us first show that $\ket{a}$ is an eigenvector of $\mathbf{A}$.
Indeed,
\bea
\mathbf{A}\ket{a}
&=&\mathbf{A}\sum_k\alpha_k\ket{k}\nonumber\\
&=&
\sum_{i,j,k}\alpha_i\alpha_j^*\alpha_k
\ket{i}\langle j|k\rangle
-\frac{1}{d}\sum_{i,k}\ket{i}\langle i|k\rangle\nonumber\\
&=&
\sum_{i,j}\alpha_i|\alpha_j|^2\ket{i}
-\frac1d\sum_i\alpha_i\ket{i}\nonumber\\
&=&
\left(1-\frac1d\right)\ket{a},
\eea
where we used the normalization $\sum_j|\alpha_j|^2=1$.
Thus, $\ket{a}$ is an eigenvector of $\mathbf{A}$ with eigenvalue
$\lambda_1=1-\frac1d$.

Next, consider any other eigenvector $\ket{\Tilde{a}}$ orthogonal to
$\ket{a}$. Then,
\bea
\mathbf{A}\ket{\Tilde{a}}
&=&
\ket{a}\langle a|\Tilde{a}\rangle
-\frac1d\ket{\Tilde{a}}\nonumber\\
&=&
-\frac1d\ket{\Tilde{a}},
\eea
which shows that $\ket{\Tilde{a}}$ is an eigenvector with eigenvalue
$\lambda_2=\cdots=\lambda_d=-\frac1d$.

Consequently, the trace norm evaluates to
\bea\label{rho_I_d}
||\rho-\I_d/d||
&=&
1-\frac1d+(d-1)\left|-\frac1d\right|\nonumber\\
&=&
2\left(1-\frac1d\right).
\eea
Using \eqref{D_S1} together with \eqref{phi}, the distinguishability of the
two depolarizing channels is therefore given by
\bea\label{const}
&&
\mathcal{D}_S\left[\{\Phi_{q_1},\Phi_{q_2}\},\{1/2,1/2\}\right]\nonumber\\
&=&
\frac12+\frac14|q_1-q_2|
2\left(1-\frac1d\right)\nonumber\\
&=&
\frac12\left(1+|q_1-q_2|\left(1-\frac1d\right)\right).
\eea
We observe that the above expression is independent of the input state
$\rho$, which completes the proof. In particular, the distinguishability
takes a constant value. 
\end{proof}
 Now we will check the necessity of auxiliary system in distinguishability of two depolarizing channels. We already know that product states do not change the scenario. Consequently, no separable state acts better than the single system probe in distinguishing two depolarizing channels for the convexity of trace norm. Now we will move to the analysis with entangled probe and eventually comment on the optimum probe for distinguishing these channels.
 \begin{lemma}
All maximally entangled states are equivalent for distinguishing two depolarizing channels. 
 \end{lemma}
 \begin{proof}
Let us consider a general maximally entangled state
$\rho^E=\ket{\phi^+}\bra{\phi^+}$, where
\[
\ket{\phi^+}=\frac{1}{\sqrt{d}}\sum_{i=0}^{d-1}\ket{i}_A\ket{i}_B .
\]
Acting with the depolarizing channel on subsystem $A$, we obtain
\[
(\Phi_{q_1}\otimes\I_d)\rho^E
=
q_1\rho^E+(1-q_1)\frac{\I_d}{d}\otimes\tr_A\rho^E .
\]
For any maximally entangled state, $\tr_A\rho^E=\I_d/d$. Hence, using
\eqref{phi}, we can write
\bea
&& ||(\Phi_{q_1}\otimes\I_d)\rho^E-(\Phi_{q_2}\otimes\I_d)\rho^E|| \nonumber\\
&=& |q_1-q_2|
.||\hspace{0.05 cm}
\ket{\phi^+}\bra{\phi^+}
-\frac{\I_d}{d}\otimes\frac{\I_d}{d}
||\nonumber\\
&=& |q_1-q_2|||\mathbf{B}|| .
\eea
It is straightforward to verify that $\ket{\phi^+}$ is an eigenvector of
$\mathbf{B}$.
\bea
\mathbf{B}\ket{\phi^+}
&=&
\left(1-\frac{1}{d^2}\right)\ket{\phi^+}.
\eea
Let $\ket{\phi'}$ be any other eigenvector orthogonal to $\ket{\phi^+}$. Then,
\bea
\mathbf{B}\ket{\phi'}
&=&
-\frac{1}{d^2}\ket{\phi'}.
\eea
Therefore, the eigenvalues of $\mathbf{B}$ are
$(1-\frac{1}{d^2})$ with multiplicity $1$ and
$-\frac{1}{d^2}$ with multiplicity $(d^2-1)$.
Consequently,
\bea\label{B}
||\mathbf{B}||
&=&
1-\frac{1}{d^2}
+(d^2-1)\left|-\frac{1}{d^2}\right|\nonumber\\
&=&
2\left(1-\frac{1}{d^2}\right).
\eea
Thus, the distinguishability of the two depolarizing channels in this
scenario is given by
\bea\label{const_E}
&&
\mathcal{D}_{ME}\left[\{\Phi_{q_1},\Phi_{q_2}\},\{1/2,1/2\}\right]\nonumber\\
&=&
\frac12+\frac14|q_1-q_2|
2\left(1-\frac{1}{d^2}\right)\nonumber\\
&=&
\frac12\left(1+|q_1-q_2|\left(1-\frac{1}{d^2}\right)\right),
\eea
which is again a constant. The suffix `ME' denotes the maximally entangled
state.
 \end{proof}
 From the last two theorems, it is evident that a maximally entangled probe outperforms a single-system probe. This naturally raises the question of whether the maximally entangled state is the optimal probe. We answer this question in the affirmative for the case of two depolarizing channels acting on $\mathbbm{C}^2$.
\begin{thm}
    The optimum probe for distinguishing two qubit depolarizing channels is maximally entangled probe.
\end{thm}
\begin{proof}
    We have to prove that maximally entangled state is strictly better probe than the non-maximally entangled state as we already know maximally entangled probes are better than single system. Let us take a general qubit-qubit non-maximally entangled probe $\ket{\phi^-}=\sqrt{g}\ket{s}\ket{t}+e^{\mathbbm{i}z}\sqrt{1-g}\ket{s^\perp}\ket{t^\perp}$. With this probe,
 \bea
&& ||(\Phi_{q_1}\otimes\I_d)\ket{\phi^-}\bra{\phi^-}-(\Phi_{q_2}\otimes\I_d)\ket{\phi^-}\bra{\phi^-}\hspace{0.05 cm}|| \nonumber\\
&=& |q_1-q_2|.||\ket{\phi^-}\bra{\phi^-}-\I_2/2\otimes\tr_A\ket{\phi^-}\bra{\phi^-}\hspace{0.05 cm}||\nonumber\\
&=&|q_1-q_2|.|| \mathbf{F}||.
    \eea
Here $\tr_A\ket{\phi^-}\bra{\phi^-}=g\ket{s}\bra{s}+(1-g)\ket{t}\bra{t}$. So $\mathbf{F}= 
\begin{pmatrix}
g-\frac{g}{2} & 0 & 0 & e^{-\mathbbm{i}z}\sqrt{g(1-g)} \\
0 & -\frac{1-g}{2} & 0 & 0 \\
0 & 0 & -\frac{g}{2} & 0\\
e^{\mathbbm{i}z}\sqrt{g(1-g)} & 0 & 0 & (1-g)-\frac{1-g}{2}
\end{pmatrix}
.$
The eigenvalues of the above matrix are $-\frac{1-g}{2}, -\frac{g}{2}, \frac{1\pm \sqrt{1+12g(1-g)}}{4}$. Consequently,
\bea
||\mathbf{F}||&=& \frac12 +\frac12\sqrt{1+12g(1-g)}.
\eea
The value of $||\mathbf{F}||$ is maximum when $g=1/2$, which indicates the probe is maximally entangled. One can cross-check that $||\mathbf{F}||=3/2$ at $g=1/2$ and this is same as the value of $||\mathbf{B}||$ of \eqref{B} with $d=2$. 
\end{proof}
Moving on, we examine the usefulness of entanglement in comparison to its absence. In the following lemma, we show that entangled states—irrespective of their degree of entanglement—are always advantageous over single-system probes in distinguishing two qubit depolarizing channels.
\begin{thm}
    Any entangled probe gives better distinguishing probability between any two qubit depolarizing channels with respect to single system probe.
\end{thm}
\begin{proof}
    From \eqref{rho_I_d}, we get $||\rho-\I_2/2||=1 \leqslant \frac12 +\frac12\sqrt{1+12g(1-g)}=||\mathbf{B}||$. The equality is achieved when $g=0$ or $1$, which means the probe is a product state. That completes our proof.  
\end{proof}
\subsection{Dephasing channels}
Generally dephasing channels are known as 'phase-flip channels'. This channel acts as follows on any given
density operator:
\bea
\mathcal{G}_r(\rho)=r\rho+(1-r) Z\rho Z;
\eea
$Z$ is a $d$-dimensional phase-flip gate, which can be written as $Z=\sum_{k=0}^{d-1}\omega^k\ket{k}\bra{k}$, where $\omega=e^{\mathbbm{i}2\pi/d}$. We want to examine the hierarchy of the probes in the discrimination of two dephasing channels with noise parameter $r_1$ and $r_2$.
\begin{thm}
    For distinguishing two dephasing channels optimally, single system probe is sufficient.
\end{thm}
\begin{proof}
Consider the optimal single-system probe state $\tau$. Under the action of two
dephasing channels, we obtain
\bea
||\mathcal{G}_{r_1}(\tau)-\mathcal{G}_{r_2}(\tau)|| 
&=& 
|r_1-r_2|\,||\tau-Z\tau Z|| .
\eea

Following the approach of Lemma \ref{th1}, let us take $\tau=\ket{b}\bra{b}$
with $\ket{b}=\sum_k\beta_k\ket{k}$. Then,
\be
\tau-Z\tau Z
=
\sum_{k,l=0}^{d-1}
\left(
\beta_k\beta_l^*\ket{k}\bra{l}
-
\beta_k\beta_l^*\omega^k(\omega^l)^*\ket{k}\bra{l}
\right)
=
\mathbf{C}.
\ee
Our task is to determine the eigenvalues of $\mathbf{C}$. Observe that
\[
\mathbf{C}
=
\ket{b}\bra{b}
-
\ket{b'}\bra{b'},
\qquad
\text{where }
\ket{b'}=\sum_k\omega^k\beta_k\ket{k}.
\]
The rank of $\mathbf{C}$ is at most $2$, and thus all non-zero eigenvalues lie
in the subspace spanned by $\{\ket{b},\ket{b'}\}$. We may express
\[
\ket{b'}
=
c\ket{b}
+
\sqrt{1-|c|^2}\ket{b^\perp},
\]
where $\ket{b^\perp}$ is orthogonal to $\ket{b}$. In the basis
$\{\ket{b},\ket{b^\perp}\}$, the operator $\mathbf{C}$ takes the matrix form
\bea
\mathbf{C}
&=&
\begin{pmatrix}
1 - |c|^2 & -c^{*}\sqrt{1 - |c|^2} \\
- c\sqrt{1 - |c|^2} & -(1 - |c|^2)
\end{pmatrix}.
\eea

The eigenvalues of this matrix are
\bea
\lambda_{\pm}
&=&
\pm\sqrt{1-|c|^2}\nonumber\\
&=&
\pm\sqrt{1-|\langle b|b' \rangle|^2}\nonumber\\
&=&
\pm\sqrt{1-\left|\sum_k |\beta_k|^2\omega^k\right|^2}\nonumber\\
&=&
\pm\sqrt{1-\left| \text{con}\{\omega^k\}\right|^2},
\eea
where $\text{con}\{\omega^k\}$ denotes the set of complex numbers expressible
as convex combinations of $\{\omega^k\}_k$. Therefore,
\be
||\tau-Z\tau Z||
=
2\sqrt{1-\left| \text{con}\{\omega^k\}\right|^2}.
\ee

The distinguishability of two
dephasing channels yields as:
\bea\label{phasing}
&&
\mathcal{D}_S\left[\{\mathcal{G}_{r_1},\mathcal{G}_{r_2}\},\{1/2,1/2\}\right]\nonumber\\
&=&
\frac12\left(
1
+
|r_1-r_2|
\sqrt{
1-\min\left| \text{con}\{\omega^k\}\right|^2
}
\right).
\eea
Here, $\min|\cdot|$ denotes the minimum modulus over all complex numbers in
the convex hull. Since the convex hull of $\{\omega^k\}_k$ always contains the
origin in the Argand plane, we have
\[
\min\left| \text{con}\{\omega^k\}\right| = 0,
\]
which is attained by choosing $\beta_k = 1/\sqrt{d}$ for all $k$. Hence,
\eqref{phasing} simplifies to
\bea\label{phasing1}
&&
\mathcal{D}_S\left[\{\mathcal{G}_{r_1},\mathcal{G}_{r_2}\},\{1/2,1/2\}\right]\nonumber\\
&=&
\frac12\left(1+|r_1-r_2|\right).
\eea

This expression represents the optimal distinguishability of two dephasing
channels, achievable using a single-system probe.
\end{proof}
The same value of distinguishability can be achieved with the entangled probe $\frac{1}{\sqrt{d}}\sum_k\ket{kk}$, but no entangled probe can exceed the value of \eqref{phasing1}.

 One should note that the distinguishability of two dephasing channels is constant (as noise parameters are fixed) irrespective of dimension.
\subsubsection{Generalization of Dephasing Channels}
Now we can think of generalized version of dephasing channel where an arbitrary unitary can act instead of phase flip gate. So this type of channels can be described as,
\bea
\mathcal{J}_r(\rho)=r\rho+(1-r)U\rho U^\dagger.
\eea
\begin{thm}
    To distinguish two generalized dephasing channels $\mathcal{J}_{r_1}$ and $\mathcal{J}_{r_2}$, single system probe is sufficient.
\end{thm}
\begin{proof}
 The argument proceeds along the same lines as in the proof of the previous
theorem. In the present case, we need to evaluate the trace norm of
$(\tau-U\tau U^\dagger)$, which again is an operator of rank at most two.
Consequently,
\bea\label{tau_UrhoU}
|| \tau-U\tau U^\dagger||
&=&
2\sqrt{1-|\langle\tau|U\tau\rangle|^2}.
\eea

Any unitary operator $U$ admits the spectral decomposition
$U=\sum_i e^{\mathbbm{i}\Theta_i}\ket{\psi_i}\bra{\psi_i}$, and the probe state
can be expanded as $\ket{\tau}=\sum_i \gamma_i\ket{\psi_i}$. Substituting these
expressions into \eqref{tau_UrhoU}, we obtain
\bea
|| \tau-U\tau U^\dagger||
&=&
2\sqrt{1-\left|\sum_i |\gamma_i|^2 e^{\mathbbm{i}\Theta_i}\right|^2}\nonumber\\
&=&
2\sqrt{1-\left|\text{con}\{ e^{\mathbbm{i}\Theta_i}\}\right|^2},
\eea
where $\text{con}\{\cdot\}$ denotes the convex hull of the set
$\{e^{\mathbbm{i}\Theta_i}\}_i$.

The resulting distinguishability is therefore given by
\bea\label{gphasing1}
&&
\mathcal{D}_S\left[\{\mathcal{J}_{r_1},\mathcal{J}_{r_2}\},\{1/2,1/2\}\right]\nonumber\\
&=&
\frac12\left(
1
+
|r_1-r_2|
\sqrt{
1-\min\left|\text{con}\{ e^{\mathbbm{i}\Theta_i}\}\right|^2
}
\right).
\eea

The optimal distinguishability is achieved when the distance between the origin
and the convex hull with vertices $\{e^{\mathbbm{i}\Theta_i}\}_i$ is minimized.
Accordingly, the optimal single-system probe is of the form
$\sum_i\gamma_i\ket{\psi_i}$, where $\{\ket{\psi_i}\}_i$ are the eigenvectors of
$U$ associated with eigenvalues $\{e^{\mathbbm{i}\Theta_i}\}_i$. Employing an
entangled probe ultimately leads to the same expression \eqref{gphasing1}.
Hence, the use of entanglement does not enhance the distinguishability in this scenario.
\end{proof}

Another type of distinguishing task can be devised if we consider two channels with same noise parameter but with different unitary. It can be proved that optimum distinguishability can be achieved with single system probe. The proof goes along the same line. Therefore, we do not discuss that case here.
\subsection{Amplitude damping channels}
The amplitude damping channel is an approximation to a noisy evolution that occurs in many physical systems. We define this channel in qubit computational basis. Any qubit amplitude damping channel $\mathcal{H}_\mu$ acts on a density operator $\rho$ as follows:
\bea
\mathcal{H}_\mu (\rho) &=& H_0\rho H_0^\dagger +H_1\rho H_1^\dagger;
\eea
where $H_0=\ket{0}\bra{0}+\sqrt{\mu}\ket{1}\bra{1}$ and $H_1=\sqrt{1-\mu}\ket{0}\bra{1}$ with $\mu\in(0,1)$. We consider the discrimination of two channels $\mathcal{H}_{\mu_1}$ and $\mathcal{H}_{\mu_2}$. Without loss of generality, we can take $\mu_1 >\mu_2$.

    For a single system $\delta$, we can write,
\bea
&&||\mathcal{H}_{\mu_1}(\delta)-\mathcal{H}_{\mu_2} (\delta)||\nonumber\\
&=&||(\mu_1-\mu_2)\ket{1}\bra{1}\delta\ket{1}\bra{1}-(\mu_1-\mu_2)\ket{0}\bra{1}\delta\ket{1}\bra{0}\nonumber\\
&&+(\sqrt{\mu_1}-\sqrt{\mu_2})\ket{1}\bra{1}\delta\ket{0}\bra{0}\nonumber\\
&&+(\sqrt{\mu_1}-\sqrt{\mu_2})\ket{0}\bra{0}\delta\ket{1}\bra{1}||\nonumber\\
&=& ||\mathbf{L} ||.
\eea
For optimality, we can take $\delta$ a pure state, i.e., $\delta=\ket{d}\bra{d}$, where $\ket{d}=\cos\frac{\theta}{2}\ket{0}+e^{\mathbbm{i}\Delta}\sin\frac{\theta}{2}\ket{1}$. So $\mathbf{L}$ can be written as,
\begin{widetext}
\bea\label{F}
\mathbf{L}&=&
\begin{pmatrix}
    -(\mu_1-\mu_2)\sin^2\frac{\theta}{2} & (\sqrt{\mu_1}-\sqrt{\mu_2})e^{-\mathbbm{i}\Delta}\cos\frac{\theta}{2}\sin\frac{\theta}{2}\\
    (\sqrt{\mu_1}-\sqrt{\mu_2})e^{\mathbbm{i}\Delta}\cos\frac{\theta}{2}\sin\frac{\theta}{2} & (\mu_1-\mu_2)\sin^2\frac{\theta}{2}
\end{pmatrix}\nonumber\\
\eea
Consequently, 
\bea
||\mathbf{L}||=2\left[(\sqrt{\mu_1}-\sqrt{\mu_2})^2\sin^2\frac{\theta}{2}\cos^2\frac{\theta}{2}+(\mu_1-\mu_2)^2\sin^4\frac{\theta}{2}\right]^{1/2}.\nonumber\\
\eea

We maximize the value of $||\mathbf{L}||$ for $\theta$ and find that,
\[
\max_{\theta}
||\mathbf{L}||
=
\begin{cases}
2(\mu_1-\mu_2),
& (\sqrt{\mu_1}+\sqrt{\mu_2})^2 \ge 1/2,
\\[10pt]
\displaystyle
\dfrac{ (\sqrt{\mu_1}-\sqrt{\mu_2}) }{\sqrt{ 1 - (\sqrt{\mu_1}+\sqrt{\mu_2})^2 } },
& (\sqrt{\mu_1}+\sqrt{\mu_2})^2 <1/2.
\end{cases}
\]
     
\end{widetext}
The corresponding values of $\theta$ is, $\theta=\pi$ when $(\sqrt{\mu_1}+\sqrt{\mu_2})^2 \ge 1/2$ and $\theta=2\arcsin\!\left(
\sqrt{
\frac{1}
{2\big[1-(\sqrt{\mu_1}+\sqrt{\mu_2})^2\big]}
}
\right)$ when $(\sqrt{\mu_1}+\sqrt{\mu_2})^2 < 1/2$.
Consequently, when $(\sqrt{\mu_1}+\sqrt{\mu_2})^2 \ge 1/2$, distinguishability reads as,
\bea\label{mu1}
&&
\mathcal{D}_S\left[\{\mathcal{H}_{\mu_1},\mathcal{H}_{\mu_2}\},\{1/2,1/2\}\right]\nonumber\\
&=&
\frac12\left(1+
(\mu_1-\mu_2)
\right).
\eea
When $(\sqrt{\mu_1}+\sqrt{\mu_2})^2 < 1/2$, distinguishability would be,
\bea\label{mu2}
&&
\mathcal{D}_S\left[\{\mathcal{H}_{\mu_1},\mathcal{H}_{\mu_2}\},\{1/2,1/2\}\right]\nonumber\\
&=&
\frac12\left(1+
\dfrac{ (\sqrt{\mu_1}-\sqrt{\mu_2}) }{2\sqrt{ 1 - (\sqrt{\mu_1}+\sqrt{\mu_2})^2 } }
\right).
\eea
We now turn to the use of entangled probes. For complexity, we mainly restricted ourselves into maximally entangled probe. First we will show that all the maximally entangled states are equivalent resource for distinguishing two amplitude damping channels.
\begin{lemma}
    All the maximally entangled probes are equivalent in this discrimination task.
\end{lemma}
\begin{proof}
 Let us consider a general maximally entangled probe
$\delta^{E}=\ket{n}\bra{n}$, where
\[
\ket{n}
=
\frac{1}{\sqrt{2}}
\left(
\ket{\chi}_A\ket{0}_B
+
\ket{\chi^\perp}_A\ket{1}_B
\right).
\]
We choose
\[
\ket{\chi}
=
\cos\frac{v}{2}\ket{0}
+
e^{\mathbbm{i}\mathbbm{V}}\sin\frac{v}{2}\ket{1}.
\]
In this setting, we are required to evaluate the quantity
\begin{widetext}
\bea\label{X}
||\mathbf{X}||
&=&
||(\mathcal{H}_{\mu_1}\otimes \I_2)(\delta^{E})
-
(\mathcal{H}_{\mu_2}\otimes\I_2)(\delta^{E})||\nonumber\\
  &=& ||\frac{1}{2}\Big[
(\mu_2-\mu_1)\sin^2\frac{v}{2}\,\ket{00}\!\bra{00}
+(\mu_2-\mu_1)\cos^2\frac{v}{2}\,\ket{01}\!\bra{01}
+(\mu_1-\mu_2)\sin^2\frac{v}{2}\,\ket{10}\!\bra{10}
+(\mu_1-\mu_2)\cos^2\frac{v}{2}\,\ket{11}\!\bra{11}
\nonumber\\
&&+(\mu_1-\mu_2)\cos\frac{v}{2}\sin\frac{v}{2}
\big(
\ket{00}\!\bra{01}
+\ket{01}\!\bra{00}
-\ket{10}\!\bra{11}
-\ket{11}\!\bra{10}
\big)
\nonumber\\
&&+(\sqrt{\mu_1}-\sqrt{\mu_2})\sin\frac{v}{2}\cos\frac{v}{2}
\big(
e^{-i\mathbbm{V}}\ket{00}\!\bra{10}
+e^{i\mathbbm{V}}\ket{10}\!\bra{00}
\big)
+(\sqrt{\mu_1}-\sqrt{\mu_2})\sin^2\frac{v}{2}
\big(
e^{-i\mathbbm{V}}\ket{01}\!\bra{10}
+e^{i\mathbbm{V}}\ket{10}\!\bra{01}
\big)
\nonumber\\
&&-(\sqrt{\mu_1}-\sqrt{\mu_2})\cos^2\frac{v}{2}
\big(
e^{-i\mathbbm{V}}\ket{00}\!\bra{11}
+e^{i\mathbbm{V}}\ket{11}\!\bra{00}
\big)
-(\sqrt{\mu_1}-\sqrt{\mu_2})\cos\frac{v}{2}\sin\frac{v}{2}
\big(
e^{-i\mathbbm{V}}\ket{01}\!\bra{11}
+e^{i\mathbbm{V}}\ket{11}\!\bra{01}
\big)
\Big] ||\nonumber\\
&=&\frac12(\mu_1-\mu_2)+\frac12\left[(\mu_1-\mu_2)^2+4(\sqrt{\mu_1}-\sqrt{\mu_2})^2\right]^{1/2}.
  \eea    
\end{widetext}

The distinguishability in this case,
\bea\label{mu_ment}
&&
\mathcal{D}_{ME}\left[\{\mathcal{H}_{\mu_1},\mathcal{H}_{\mu_2}\},\{1/2,1/2\}\right]\nonumber\\
&=&
\frac12\left(1+
\frac14 (\mu_1-\mu_2)\right.\nonumber\\
&&+\left.\frac14\left[(\mu_1-\mu_2)^2+4(\sqrt{\mu_1}-\sqrt{\mu_2})^2\right]^{1/2}
\right).
\eea

Since the final expression does not depend on the parameter $v,\mathbbm{V}$,
different maximally entangled probes lead to same values of the trace
norm. Therefore, all maximally entangled probes are equivalent in this scenario.
\end{proof}
In this point, it is better to compare two kinds of probes in the discrimination task. We prove that maximally entangled probe is strictly worse probe with respect to single system.
\begin{thm}\label{4}
    Maximally entangled probe is strictly worse probe for the distinguishability of two qubit amplitude damping channels with respect to single system probe when $(\sqrt{\mu_1}+\sqrt{\mu_2})^2 \ge 1/2$.
\end{thm}
\begin{proof}
Suppose maximally entangled states give better probability than single systems in the given range. From \eqref{mu_ment} and \eqref{mu1}, it offers us,
\bea
\frac12\left[(\mu_1-\mu_2)^2+4(\sqrt{\mu_1}-\sqrt{\mu_2})^2\right]^{1/2} > \frac32(\mu_1-\mu_2)\nonumber\\
\eea
Simplifying this, we get
\bea
(\sqrt{\mu_1}+\sqrt{\mu_2})^2 < \frac12,
\eea
which contradicts the range of \eqref{mu1}. So maximally entangled states are strictly worse probe than single system probe in the given range.
\end{proof}
In contrast of last result, we can show that maximally entangled states can be useful as the probe with respect to single system if we restrict the noise parameters wisely.
\begin{thm}
     Maximally entangled probe is strictly better probe for distinguishability of two qubit amplitude damping channels with respect to single system probe when $(\sqrt{\mu_1}+\sqrt{\mu_2})^2 < 1/2$.
\end{thm}
\begin{proof}
    From \eqref{mu2} and \eqref{mu_ment}, one can start with that
    \bea
&&\frac12(\mu_1-\mu_2)+\frac12\left[(\mu_1-\mu_2)^2+4(\sqrt{\mu_1}-\sqrt{\mu_2})^2\right]^{1/2}\nonumber\\
&&>\dfrac{ (\sqrt{\mu_1}-\sqrt{\mu_2}) }{\sqrt{ 1 - (\sqrt{\mu_1}+\sqrt{\mu_2})^2 } }.
    \eea
We can write, $(\mu_1-\mu_2)=(\sqrt{\mu_1}+\sqrt{\mu_2})(\sqrt{\mu_1}-\sqrt{\mu_2})$. Implementing this and squaring the both sides of the above inequality, we will get,
\be
(\sqrt{\mu_1}+\sqrt{\mu_2})^2 < 1/2,
\ee
which is consistent with \eqref{mu2}.That suffices the proof.
\end{proof}

We move to the next result where the optimality of non-maximally entangled state is established.
\begin{thm}
    Non-maximally entangled probe is the optimal probe for the discrimination of two qubit amplitude damping channels if $(\sqrt{\mu_1}+\sqrt{\mu_2})^2 \ge 1/2$.
\end{thm}
\begin{proof}
    In Theorem \ref{4}, it has been proved that single system is strictly better probe than maximally entangled probe if $(\sqrt{\mu_1}+\sqrt{\mu_2})^2 \ge 1/2$. Here, it would be sufficient if we show that non-maximally entangled is strictly better probe than single system probe in the given range.  We take a non-maximally entangled probe $\delta'=\ket{\phi^{''}}\bra{\phi^{''}}$, where $\ket{\phi^{''}}=\sqrt{p}\ket{00}+\sqrt{1-p}\ket{11}$, where $p\neq 0,1$. Using this probe, we calculate the following quantity,
\bea\label{N}
||\mathbf{N}||
&=&
||(\mathcal{H}_{\mu_1}\otimes \I_2)(\delta')
-
(\mathcal{H}_{\mu_2}\otimes\I_2)(\delta')||\nonumber\\
  &=&  || (\mu_1-\mu_2)(1-p)\ket{11}\bra{11}\nonumber\\
  &&-(\mu_1-\mu_2)(1-p)\ket{01}\bra{01}\nonumber\\
  &&+\sqrt{p(1-p)}(\sqrt{\mu_1}-\sqrt{\mu_2})(\ket{00}\bra{11}+\ket{11}\bra{00})\nonumber\\
  &=& (1-p)(\mu_1-\mu_2)\nonumber\\
  &&+(1-p)\left[(\mu_1-\mu_2)^2+4\frac{p}{1-p}(\sqrt{\mu_1}-\sqrt{\mu_2})^2\right]^{1/2}.\nonumber\\
  \eea 
  We have to show,
  \[
  ||\mathbf{N}||>||\mathbf{L}||
  \]
  which necessitates,
    \bea
&&(1-p)\Bigg\{(\mu_1-\mu_2)\nonumber\\
&&+\left[(\mu_1-\mu_2)^2+4\frac{p}{1-p}(\sqrt{\mu_1}-\sqrt{\mu_2})^2\right]^{1/2}\Bigg\}\nonumber\\
&& > 2(\mu_1-\mu_2)
    \eea
Simplifying this, we get 
\bea
(\sqrt{\mu_1}+\sqrt{\mu_2})^2< 1-p;
\eea
In the necessary range, the above expression can be reduced as,
\bea
1-p\ge\frac12,
\eea
which gives, $p\le \frac12$ ($p=1/2$ corresponds to maximally entangled state). If the values of $p$ satisfies $p<1/2$, non-maximally entangled state serves as the best probe in the given range of noise parameters.
\end{proof}
In the previous proof, we do not comment about the hierarchy between the non-maximally entangled states. We just provide the existence of a non-maximally entangled probe which is more successful in that task with respect to both single system and maximally entangled system.

Due to the complexity of calculation, we could not find the optimal probe for the case $(\sqrt{\mu_1}+\sqrt{\mu_2})^2 \ge 1/2$. 
\subsection{Noisy unitary channels}
The noisy unitary channel $\mathcal{U}$ is defined as an ensemble of unitaries $\{U_k,q_k\}_k$, where each unitary $U_k$ is sampled with probability $q_k$. This channel acts on a density operator $\rho$ and gives the output as,
\bea
\mathcal{U}(\rho)=\sum_k q_k U_k\rho U_k^\dagger.
\eea
Firstly, we consider the distinguishability of two noisy unitaries $\mathcal{V}$ and $\mathcal{W}$, which are defined as $\mathcal{V}=\{V_k,q_k\}$ and $\mathcal{W}=\{W_k,q_k\}$. For sake of simplicity, we take the sampling probability to be same.

    With single system probe $\delta$, the distinguishability depends on the quantity $||\mathcal{V}(\delta)-\mathcal{W}(\delta)||$. 
\bea\label{mix_u}
||\mathcal{V}(\delta)-\mathcal{W}(\delta)|| &=& ||\sum_k q_k V_k\rho V_k^\dagger- \sum_k q_k W_k\rho W_k^\dagger ||\nonumber\\
&\leqslant & \sum_k || q_k \left(V_k\rho V_k^\dagger-  W_k\rho W_k^\dagger\right) ||\nonumber\\
&=& \sum_k q_k ||V_k\rho V_k^\dagger-  W_k\rho W_k^\dagger||.
\eea
If we take $\delta$ as a pure state,i.e., $\delta=\ket{d}\bra{d}$, Hellstrom formula reduces \eqref{mix_u} to,
\bea\label{nor_bound}
||\mathcal{V}(\delta)-\mathcal{W}(\delta)|| &\leqslant & 2 \sum_k q_k \sqrt{1-|\langle d|V_k^\dagger W_k|d\rangle|^2}\nonumber\\
\eea

If we start with maximally entangled state $\ket{\phi^+}=\frac{1}{\sqrt{d}}\sum_i\ket{ii}$, similarly the above bound takes form as following:

\bea\label{ent_bound}
&&||(\mathcal{V}\otimes\I_d)(\ket{\phi^+}\bra{\phi^+})-(\mathcal{W}\otimes\I_d)(\ket{\phi^+}\bra{\phi^+})|| \nonumber\\&\leqslant & 2 \sum_k q_k \sqrt{1-|\langle \phi^+|(V_k^\dagger\otimes\I_d) (W_k\otimes\I_d)|\phi^+\rangle|^2}\nonumber\\
&=&  2 \sum_k q_k \sqrt{1-\frac{1}{d^2}|\sum_{i=1}^d\bra{i}V_k^\dagger W_k\ket{i}|^2}\nonumber\\
&=& 2 \sum_k q_k \sqrt{1-\frac{1}{d^2}|\tr(V_k^\dagger W_k)|^2}.
\eea
As trace is a basis-independent property, the last line comes trivially after the third line. For any of the pairs $\{V_k, W_k\}$, if $\tr(V_k^\dagger W_k)\neq 0$, the mixed unitaries are not distinguishable with maximally entangled probe.

 If $q_k=1$, i.e., the unitaries are not mixed, the single system probe is sufficient \cite{bsxv-q9x7}. But this is not the case for mixed unitaries. Let us take the two noisy unitaries $\mathcal{L}=\{L_k,q_k\}_{k=1}^3$ and $\mathcal{S}=\{S_k,q_k\}_{k=1}^3$ in dimension $3$ and $q_k \neq 0,1, \forall k$. The components are defined as follows:
\begin{align}
L_1 &= \ket{0}\bra{0} + \ket{1}\bra{1} + \ket{2}\bra{2}, \nonumber\\
L_2 &= \ket{1}\bra{0} + \ket{2}\bra{1} + \ket{0}\bra{2}, \nonumber\\
L_3 &= -\ket{0}\bra{0} + \ket{1}\bra{1} + \ket{2}\bra{2}, \nonumber\\
S_1 &= \ket{2}\bra{0} + \ket{0}\bra{1} + \ket{1}\bra{2}, \nonumber\\
S_2 &= -\ket{1}\bra{0} + \ket{2}\bra{1} + \ket{0}\bra{2}, \nonumber\\
S_3 &= -\ket{2}\bra{0} + \ket{0}\bra{1} + \ket{1}\bra{2}.
\end{align}


\begin{thm}
    Non-maximally entangled state is the optimal probe for the discrimination of two noisy unitary channels $\mathcal{L}$ and $\mathcal{S}$.
\end{thm}
\begin{proof}
   First, we consider a non-maximally entangled probe
$Z_{AB}=\ket{\zeta}\bra{\zeta}$, where
\[
\ket{\zeta}=\frac{1}{\sqrt{2}}\ket{00}+c_1\ket{11}+c_2\ket{22}.
\]
Normalization implies that $|c_1|^2+|c_2|^2=1/2$. Consequently,
\bea\label{non_max}
&&||(\mathcal{L}\otimes\I_3)(Z)
-(\mathcal{S}\otimes\I_3)(Z)||\nonumber\\
&=& \Big\|\sum_k q_k (L_k\otimes\I_3)Z(L_k^\dagger\otimes\I_3)
-\sum_k q_k (S_k\otimes\I_3)Z(S_k^\dagger\otimes\I_3)\Big\|\nonumber\\
&=& \Big\| q_1 (\ket{\zeta_1}\bra{\zeta_1}-\ket{\zeta_4}\bra{\zeta_4})
+q_2 (\ket{\zeta_2}\bra{\zeta_2}-\ket{\zeta_5}\bra{\zeta_5})\nonumber\\
&&\qquad
+q_3 (\ket{\zeta_3}\bra{\zeta_3}-\ket{\zeta_6}\bra{\zeta_6})\Big\|\nonumber\\
&=&||\mathbf{Y}||,
\eea
where
\[
\begin{aligned}
&\ket{\zeta_1}=\frac{1}{\sqrt{2}}\ket{00}+c_1\ket{11}+c_2\ket{22},\\
&\ket{\zeta_2}=\frac{1}{\sqrt{2}}\ket{10}+c_1\ket{21}+c_2\ket{02},\\
&\ket{\zeta_3}=-\frac{1}{\sqrt{2}}\ket{00}+c_1\ket{11}+c_2\ket{22},\\
&\ket{\zeta_4}=\frac{1}{\sqrt{2}}\ket{20}+c_1\ket{01}+c_2\ket{12},\\
&\ket{\zeta_5}=-\frac{1}{\sqrt{2}}\ket{10}+c_1\ket{21}+c_2\ket{02},\\
&\ket{\zeta_6}=-\frac{1}{\sqrt{2}}\ket{20}+c_1\ket{01}+c_2\ket{12}.
\end{aligned}
\]
One can verify that the set of states $\{\ket{\zeta_k}\}_k$ is an
orthogonal set. So these states are the eigenvectors of $||\mathbf{Y}||$. From~\eqref{non_max}, we obtain
\bea
&&||(\mathcal{L}\otimes\I_3)\mathbf{Z}(\mathcal{L}^\dagger\otimes\I_3)
-(\mathcal{S}\otimes\I_3)\mathbf{Z}(\mathcal{S}^\dagger\otimes\I_3)||\nonumber\\
&=& |q_1|+|-q_1|+|q_2|+|-q_2|+|q_3|+|-q_3|\nonumber\\
&=& 2.
\eea
This leads to
\bea
&&\mathcal{D}_E\!\left[\{\mathcal{L},\mathcal{S}\},\{1/2,1/2\}\right]\nonumber\\
&=& \frac12+\frac14
||(\mathcal{L}\otimes\I_3)\mathbf{Z}-(\mathcal{S}\otimes\I_3)\mathbf{Z}||\nonumber\\
&=& 1.
\eea

If we are restricted to a single-system probe, the distinguishability is
strictly less than $1$. To show this, we invoke the bound in~\eqref{nor_bound},
which in the present case reduces to
\bea
||\mathcal{L}(\delta)-\mathcal{S}(\delta)|| 
&\leqslant&
2 \sum_k q_k \sqrt{1-|\langle d|L_k^\dagger S_k|d\rangle|^2}.
\eea
Since the probabilities $q_k$ are fixed, the right-hand side would attain the
value $2$ only if all the terms
$\{\sqrt{1-|\langle d|L_k^\dagger S_k|d\rangle|^2}\}_k$
are equal to $1$. This requires a common probe state $\ket{d}$ that perfectly
distinguishes all the pairs $\{L_k,S_k\}_k$.

First, consider the pair $\{L_2,S_2\}$. Let
$\ket{d}=\sum_{i=0}^2 d_i\ket{i}$. We compute
\bea
|\langle d|L_2^\dagger S_2|d\rangle|
&=& \left||d_0|^2(-1)+|d_1|^2(1)+|d_2|^2(1)\right|.
\eea
This expression vanishes if $|d_0|^2=1/2$ and
$|d_1|^2+|d_2|^2=1/2$. Hence the probe can be written as
\[
\ket{d}=\frac{1}{\sqrt{2}}\ket{0}+d_1\ket{1}+d_2\ket{2}.
\]
With this choice, perfect distinguishability of $L_1$ and $S_1$ yields
\bea\label{cond1}
\frac{d_2}{\sqrt{2}}+\frac{d_1^*}{\sqrt{2}}+d_2^*d_1=0.
\eea
Similarly, distinguishability of $L_3$ and $S_3$ gives
\bea\label{cond2}
-\frac{d_2}{\sqrt{2}}-\frac{d_1^*}{\sqrt{2}}+d_2^*d_1=0.
\eea
Adding~\eqref{cond1} and~\eqref{cond2}, we obtain
\be
d_2^*d_1=0.
\ee
This implies $d_2=0$ or $d_1=0$. Substituting either choice into
\eqref{cond1} or \eqref{cond2} forces $d_2=d_1=0$, contradicting normalization.
Therefore, no common probe exists that can distinguish all three pairs, and
hence
\bea
||\mathcal{L}(\delta)-\mathcal{S}(\delta)|| < 2.
\eea

One can check that $\tr(L_2^\dagger S_2)=1\neq 0$. From\eqref{ent_bound}, we can infer that,
\bea
||(\mathcal{L}\otimes\I_3)(\ket{\phi^+}\bra{\phi^+})-(\mathcal{S}\otimes\I_3)(\ket{\phi^+}\bra{\phi^+})|| < 2.
\eea
That completes our proof.
\end{proof}
Now we will show that there exists two noisy unitary channels, where the maximally entangled state acts as the strictly worse probe than the single system probe. To this note, let us consider two noisy unitaries $\mathcal{\overline{L}}=\{\overline{L}_k,q_k\}_{k=1}^3$ and $\mathcal{\overline{S}}=\{\overline{S}_k,q_k\}_{k=1}^3$ in dimension $6$ and $q_k \neq 0,1, \forall k$. The components are defined as follows:
\begin{align}
\overline{L}_1 &= \ket{0}\bra{0}+\ket{1}\bra{1}+\ket{2}\bra{2}+\ket{3}\bra{3}+\ket{4}\bra{4}+\ket{5}\bra{5},\nonumber\\
\overline{L}_2 &= \ket{1}\bra{0}+\ket{0}\bra{1}+\ket{2}\bra{2}+\ket{3}\bra{3}+\ket{4}\bra{4}+\ket{5}\bra{5}, \nonumber\\
\overline{L}_3 &= \ket{2}\bra{0}+\ket{1}\bra{1}+\ket{0}\bra{2}+\ket{3}\bra{3}+\ket{4}\bra{4}+\ket{5}\bra{5}, \nonumber\\
\overline{S}_1 &= \ket{3}\bra{0}+\ket{1}\bra{1}+\ket{2}\bra{2}+\ket{0}\bra{3}+\ket{4}\bra{4}+\ket{5}\bra{5}, \nonumber\\
\overline{S}_2 &= \ket{4}\bra{0}+\ket{1}\bra{1}+\ket{2}\bra{2}+\ket{3}\bra{3}+\ket{0}\bra{4}+\ket{5}\bra{5}, \nonumber\\
\overline{S}_3 &= \ket{5}\bra{0}+\ket{1}\bra{1}+\ket{2}\bra{2}+\ket{3}\bra{3}+\ket{4}\bra{4}+\ket{0}\bra{5}.
\end{align}
\begin{thm}
    Single system is the better probe for the discrimination of two noisy unitary channels $\mathcal{\overline{L}}$ and $\mathcal{\overline{S}}$ than maximally entangled system.
\end{thm}
\begin{proof}
    It is easy to check that $\tr(\overline{L}_i^\dagger\overline{S}_i)\neq 0$ for all $i=1,2,3$. From \eqref{ent_bound}, we deduce that maximally entangled probe can not discriminate two channels $\mathcal{\overline{L}}$ and $\mathcal{\overline{S}}$.   

    As a single system probe, we choose $\ket{0}\bra{0}$. Consequently,
    \bea\label{singles}
&&||\mathcal{\overline{L}}(\ket{0}\bra{0})
-\mathcal{\overline{S}}(\ket{0}\bra{0})||\nonumber\\
&=& \Big\|\sum_k q_k \overline{L}_k\ket{0}\bra{0}\overline{L}_k^\dagger -\sum_k q_k \overline{S}_k\ket{0}\bra{0}\overline{S}_k^\dagger\Big\|\nonumber\\
&=& \Big\| q_1 (\ket{0}\bra{0}-\ket{3}\bra{3})
+q_2 (\ket{1}\bra{1}-\ket{4}\bra{4})\nonumber\\
&&\qquad
+q_3 (\ket{2}\bra{2}-\ket{5}\bra{5})\Big\|\nonumber\\
&=& |q_1|+|-q_1|+|q_2|+|-q_2|+|q_3|+|-q_3|\nonumber\\
&=& 2
\eea
So the unitaries are perfectly distinguished using single system.
\end{proof}
For mixed unitaries with more intricate structures, one may allow unequal sampling probabilities. Such scenarios, however, lie beyond the scope of the present work.
\subsection{Erasure Channels}
A quantum erasure channel either transmits the input state with some probability $\epsilon$ or replaces it with some state $\ket{e}$ with some probability $(1-\epsilon)$, where $\ket{e}$ is not in the input Hilbert space, which means $\ket{e}$ is orthogonal to the input state. It implements the action as follows:
\be
\mathcal{E}_{\epsilon}(\rho)= \epsilon\rho + (1-\epsilon)\ket{e}\bra{e}.
\ee
\begin{lemma}
    All pure single system probes are equivalent for distinguishing two eraser channels.
\end{lemma}
\begin{proof}
    With a single system probe $\rho$, the necessary trace norm term will be,
    \bea
||\mathcal{E}_{\epsilon_1}(\rho)-\mathcal{E}_{\epsilon_2}(\rho)||&=& |\epsilon_1-\epsilon_2|\hspace{0.02 cm}||\rho-\ket{e}\bra{e}|| \nonumber\\
&=&||\mathbf{W}||.
    \eea
As $\ket{e}\bra{e}$ is orthogonal to $\rho$, $\ket{e}$ is an eigenvector of $\mathbf{W}$. $\rho$ has all non-zero eigenvalues sum up to $1$. So, trace norm of the quantity is $2$.
\be\label{w}
||\mathbf{W}||= 2|\epsilon_1-\epsilon_2|
\ee
The distinguishability in this case,
\bea
&&\mathcal{D}_S\left[\{\mathcal{E}_{\epsilon_1},\mathcal{E}_{\epsilon_2}\},\{1/2,1/2\}\right]\nonumber\\
&=&\frac12+\frac12 |\epsilon_1-\epsilon_2|.
\eea
It is needless to say that this value does not depend on $\rho$. 
\end{proof}
\begin{thm}
    Single system probe is sufficient for distinguishing two erasure channels.
\end{thm}
\begin{proof}
We start with the entangled probe $\ket{\zeta^+}=\sum_i\Omega_i\ket{\kappa_i}\ket{i}$. An erasure channel acts on the first bit of the probe resulting to the transformed state $\epsilon\ket{\zeta^+}\bra{\zeta^+}+(1-\epsilon)|E\rangle\langle E|\otimes\tr_A \ket{\zeta^+}\bra{\zeta^+}$. Similarly,
\bea
&&||\mathcal{E}_{\epsilon_1}(\ket{\zeta^+}\bra{\zeta^+})-\mathcal{E}_{\epsilon_2}(\ket{\zeta^+}\bra{\zeta^+})||\nonumber\\
&=& |\epsilon_1-\epsilon_2|\hspace{0.02 cm}||\ket{\zeta^+}\bra{\zeta^+}-|E\rangle\langle E|\otimes\tr_A \ket{\zeta^+}\bra{\zeta^+}|| \nonumber\\
&=&||\mathbf{W'}||.
    \eea
From the definition, we know that $\langle \kappa_i|E\rangle=0$ for all $i$. So $|E\rangle\langle E|\otimes\tr_A \ket{\zeta^+}\bra{\zeta^+}$ is an eigenvector of $\mathbf{W'}$. In the same logic of previous lemma, we can say that ,
\be\label{w'}
||\mathbf{W'}||= 2|\epsilon_1-\epsilon_2|
\ee
The distinguishability in this case,
\bea
&&\mathcal{D}_E\left[\{\mathcal{E}_{\epsilon_1},\mathcal{E}_{\epsilon_2}\},\{1/2,1/2\}\right]\nonumber\\
&=&\frac12+\frac12 |\epsilon_1-\epsilon_2|,
\eea
which is same as $\mathcal{D}_S\left[\{\mathcal{E}_{\epsilon_1},\mathcal{E}_{\epsilon_2}\}\right]$. Therefore we can conclude the sufficiency of single system probe in this distinguishing task. 
\end{proof}
A key insight from our results is that the distinguishability of two erasure channels is independent of the choice of probing state.
\section{Conclusion}
 In this paper, we identify optimal probing states for the discrimination of two noisy quantum channels and investigate the necessity of entanglement-assisted ancillary systems in this task. We prove that a maximally entangled state serves as the optimal probe for distinguishing two qubit depolarizing channels. We further show that entanglement is not required for the discrimination of two dephasing channels or two erasure channels. Additionally, we demonstrate that, depending on the noise parameters, a single-system probe and non-maximally entangled probe can outperform a maximally entangled probe in the discrimination of two qubit amplitude-damping channels. 
For the discrimination of two noisy unitary channels, we present an explicit example in which a non-maximally entangled state outperforms both single-system and maximally entangled probes. Furthermore, we identify a class of mixed-unitary channels for which a single-system probe provides an advantage over a maximally entangled probe in the channel discrimination task.

Our investigation raises several intriguing open questions. Most notably, identifying optimal probes for the discrimination of more than two quantum channels—and consequently deriving tight upper bounds on distinguishability for a general set of channels—remains an open problem. In several scenarios considered in this work, including qudit depolarizing channels, amplitude-damping channels for some classes, and general mixed-unitary channels, we are unable to fully characterize the optimal initial probe state, highlighting the need for further study.
Another interesting direction concerns the discrimination between different types of quantum channels. For instance, one may consider distinguishing a depolarizing channel from a dephasing channel. When the noise parameters of the two channels are identical and the system dimension is restricted to two, a maximally entangled state can be shown to be the optimal probe for this task. Since the proof closely follows the arguments used in the discrimination of two qubit depolarizing channels, we omit it from the present work.
Finally, channel distinguishability naturally arises in a wide range of communication complexity problems and quantum computation theory, indicating numerous potential applications of our framework. We hope that this work provides a broader perspective on quantum channel discrimination and stimulates further research in this direction.
\subsection*{Acknowledgement}
The author thanks Debashis Saha for fruitful discussion.
\bibliography{bib}
\end{document}